\documentclass[10pt, conference]{IEEEtran}

\usepackage{url}
\usepackage{cite}

\usepackage{amsmath,amssymb,amsfonts} %amsmath already included
\usepackage{graphicx} %graphicx already included
\usepackage{subcaption}
\usepackage{textcomp}
\usepackage{xcolor}

\usepackage{booktabs}
\usepackage{makecell} 

\usepackage{tikz}
\usetikzlibrary{arrows,decorations.markings, positioning}
\usepackage{tabu}
\usepackage{mathtools}
\usepackage[normalem]{ulem}
\usepackage[switch]{lineno}

\usepackage{multirow}

\usepackage[]{scheduling}
\usepackage{colonequals} % \ratio

\usepackage{epstopdf}

\makeatletter
\DeclareRobustCommand{\cev}[1]{%
	{\mathpalette\do@cev{#1}}%
}
\newcommand{\do@cev}[2]{%
	\vbox{\offinterlineskip
		\sbox\z@{$\m@th#1 x$}%
		\ialign{##\cr
			\hidewidth\reflectbox{$\m@th#1\vec{}\mkern4mu$}\hidewidth\cr
			\noalign{\kern-\ht\z@}
			$\m@th#1#2$\cr
		}%
	}%
}

\makeatother

\DeclareMathOperator*{\argmin}{arg\,min}

\usepackage{amsthm}

\newtheorem{thm}{Theorem}%[section]

\newtheorem{lemma}[thm]{Lemma} %this used to be \newtheorem{lem}[thm]{Lemma}, which i think was unintended

\theoremstyle{definition}

\theoremstyle{remark}

\title{Optimizing Age-of-Information in Piggyback Networks with Recurrent Data Generation}

\author{
	\IEEEauthorblockN{%
		Ching-Chi Lin\IEEEauthorrefmark{1},
		Mario G\"{u}nzel\IEEEauthorrefmark{1},		
		and Jian-Jia Chen\IEEEauthorrefmark{1}\IEEEauthorrefmark{2}
	}
	\IEEEauthorblockA{%
		\IEEEauthorrefmark{1} TU Dortmund University, Dortmund, Germany\\		
		\IEEEauthorrefmark{2} Lamarr Institute for Machine Learning and Artificial Intelligence, Dortmund, Germany\\
		Email: \{chingchi.lin, mario.guenzel, jian-jia.chen\}@tu-dortmund.de
	}
}

\begin{document}

\maketitle

\begin{abstract}    		
	\emph{Age-of-information} (AoI) is a critical metric that quantifies the freshness of data in communication systems.
	In the era of the Internet of Things (IoT), data collected by resource-constrained devices often need to be transmitted to a central server to extract valuable insights in a timely manner.
	However, maintaining a stable and direct connection between a vast number of IoT devices and servers is often impractical.
	The \emph{Store-Carry-Forward} (SCF) communication paradigm, such as Piggyback networks, offers a viable solution to address the data collection and transmission challenges in distributed IoT systems by leveraging the mobility of mobile nodes.
	
	In this work, we investigate AoI within the context of \emph{patrolling data collection drones}, where data packets are generated \emph{recurrently} at devices and collected by a patrolling drone to be delivered to a server.	
	Our objective is to design a patrolling route that minimizes the \emph{Maximum Age-of-Information} (MAI) across the system.
	We demonstrate that determining whether a route with an MAI below a certain threshold can be constructed is NP-Complete.
	To address this challenge, we propose two approaches with approximation guarantees.	
	Our evaluation results show that the proposed approaches can achieve near-optimal routes in reasonable time across various scenarios.
\end{abstract}

\section{Introduction}

%%% opening
In the era of the Internet of Things (IoT), an increasing number of devices are being deployed for monitoring and data collection, generating vast amounts of information.
For instance, in large-scale environmental monitoring systems, sensors are deployed to collect data on parameters such as temperature, humidity, and air quality.
Given that many of these devices are resource-constrained, the collected data often need to be transmitted to a server for further processing and extraction of valuable insights. However, maintaining a stable and direct connection between these devices and the server is often impractical due to the high volume of data and the significant cost of infrastructure deployment.

%%% Piggyback network
Piggyback networks offer a promising solution to address data collection and transmission challenges in distributed IoT systems.
In Piggyback networks, data packets are transmitted from the source to the destination by leveraging the mobility of mobile nodes.
This approach follows the \emph{Store-Carry-Forward} (SCF) communication paradigm, where data packets are temporarily stored at intermediate nodes and then transported to their destinations by mobile nodes.
These mobile nodes can be vehicles, robots, or unmanned aerial vehicles (UAVs) equipped with high-speed storage devices and high-frequency communication modules.

%%% terahertz and mmWave
Advancements in wireless communication technologies for next-generation networks have further enhanced the viability of Piggyback networks.
High-speed, short-distance wireless communication technologies, such as \textit{Millimeter waves} (mmWave)~\cite{koenig2013wireless} and \textit{Terahertz} (THz)~\cite{akyildiz2014terahertz}, enable rapid data transfer, making them well-suited for handling large volumes of data.
In particular scenarios, Piggyback networks utilizing drones equipped with mmWave communication modules have demonstrated the potential to achieve higher throughput compared to 5G and LTE networks~\cite{SoHasegawa2023}.

%%% related works
Previous studies~\cite{9700436, SoHasegawa2023, DaisukeYamamoto20242023XBL0180} on Piggyback networks have primarily focused on throughput optimization.
These studies address scenarios where data packets need to be transported from various locations to their designated destinations.
Given that each data packet is collected and transported only once, the system's overall throughput is significantly influenced by the routing decisions made for the drones within the network.

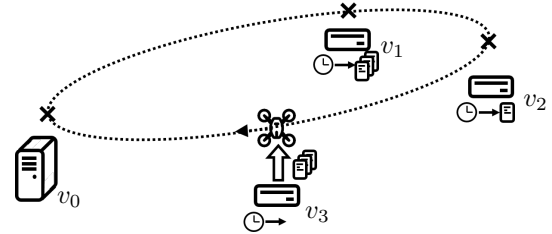
\begin{figure}[tb]
    \centering
	\begin{tikzpicture}
		\node[inner sep=0pt] (scenario) at (0,0){
        \includegraphics[width=.75\linewidth]{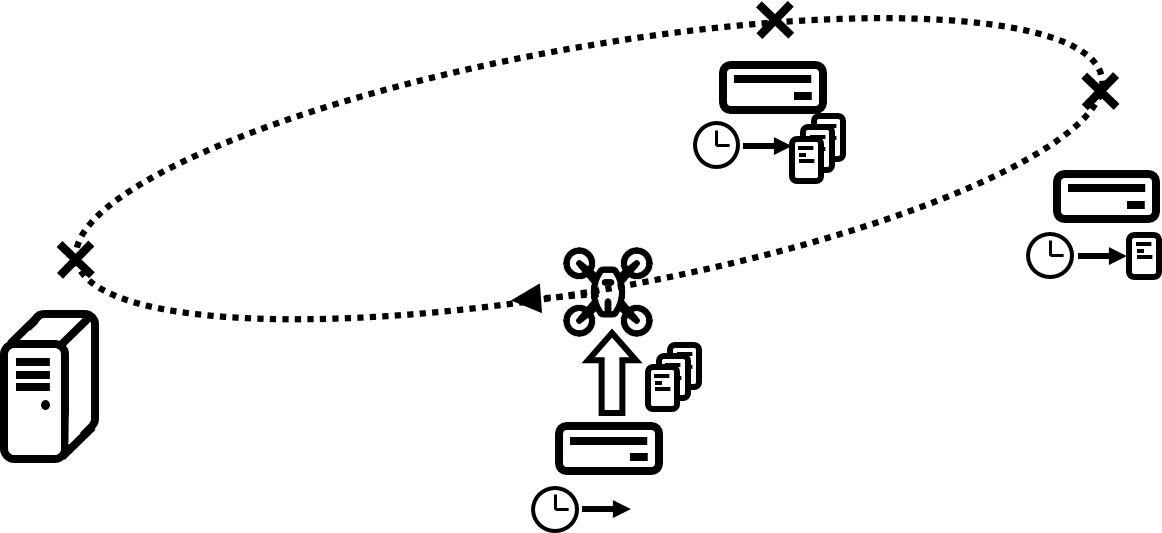}
		};
		\node[] at (-2.6,-1.1){$v_0$};
		\node[] at (1.7,0.9){$v_1$};
		\node[] at (3.6,0.2){$v_2$};
		\node[] at (0.7,-1.3){$v_3$};
	\end{tikzpicture}    
    \caption{An example of a patrolling data collection drone system with a server $v_0$ and three IoT devices $v_1$, $v_2$, and $v_3$.
	The IoT devices recurrently generate data, which the drone collects upon visiting each device.
	The drone then transports the collected data back to the server for processing.}
	\label{fig:patrolling_drone}
	\vspace{-1em}
\end{figure}

%%% gap: AoI
Instead of focusing on one-time pickup-and-delivery and optimizing throughput, this work explores an alternative application of Piggyback networks using a \textit{patrolling data collection drone} with an emphasis on \emph{Age-of-Information} (AoI).
\emph{Age-of-Information} (AoI)~\cite{Sun2019AgeOI} is a critical metric for assessing data freshness in communication systems, quantifying the duration between the generation of a data packet at the source and its arrival at the destination for processing.
AoI is highly relevant in applications where information is time-sensitive, such as vehicular networks~\cite{9896680} and disaster management~\cite{QADIR2021114}.

In the \textit{patrolling data collection drone} scenario, IoT devices are strategically deployed within a confined area, such as monitoring sensors in a forest or ocean region, or surveillance cameras around a hazardous zone.
These devices \emph{recurrently} generate data packets that must be transmitted to a server for aggregation and analysis.
While these devices generally remain stationary, they can be relocated if necessary.
Given the challenges of maintaining a direct or stable connection to the server, especially due to high data volumes or limited infrastructure, direct data transmission is often impractical.
Instead, drones are employed to patrol the area, collect data packets from these devices, and transport them back to the server for aggregation and analysis, as depicted in Figure~\ref{fig:patrolling_drone}.
In this context, minimizing the AoI is essential to ensure that the server continuously receives and processes up-to-date information from the deployed devices.

%%% contribution
In this work, we investigate the \emph{AoI-aware route planning problem} for patrolling data collection drones.
Given the locations of the server and the data nodes (i.e., the IoT devices), our objective is to design a patrolling route that minimizes the \emph{maximum Age-of-Information} (MAI) across the system.
We first demonstrate that the decision version of the \emph{AoI-aware route planning problem} is NP-Complete.
To address this challenge, we propose two approaches with approximation guarantees.

Our main contributions are as follows:

\begin{itemize}
	\item We formally define the \emph{AoI-aware route planning problem} and prove that determining whether a route with an MAI no more than a certain threshold can be constructed is NP-Complete. This proof involves a reduction from the Hamiltonian Path Problem, in Section~\ref{sec:problem_definition}.
	 	    
    \item We show that for any given patrolling route, the data packet with the maximum AoI in the worst case originates from the first data node visited.		
    Based on this observation, we propose a route design approach, \emph{Shortest Round Trip Time} (SRTT), in Section~\ref{sec:shortest_route}, and prove that \emph{SRTT} achieves an approximation ratio of $1.5$ in terms of MAI.

	\item In Section~\ref{sec:edge_enforcement}, we first demonstrate that a route with a shorter round trip time does not necessarily result in a lower MAI.
	We then propose an alternative approach, \emph{Edge Enforcement} (Enforced), which also provides an approximation ratio of $1.5$ for the MAI.

    \item Our empirical results in Section~\ref{sec:eval} demonstrate that the proposed approaches generate near-optimal routes in terms of MAI within reasonable time in various scenarios.
\end{itemize}

\section{Background and Related Works}

\subsection{Piggyback Networks}

Piggyback networks are a type of delay-tolerant network that leverages the mobility of mobile nodes to relay data packets from the source to the destination.
A key feature of Piggyback networks is the \textit{Store-Carry-Forward} (SCF) communication paradigm, where data packets are temporarily stored at intermediate nodes and then transported to their destinations by mobile nodes.
Kuwata et. al.~\cite{9700436} focused on minimizing the time required for all autonomous mobilities to finish delivering their assigned data and proposed a local search algorithm that iteratively improves data assignments.
Hasegawa et. al.~\cite{SoHasegawa2023} evaluated the data transfer performance of the Piggyback Network, demonstrating that under certain conditions, the throughput of the Piggyback Network can exceed that of traditional 5G and LTE networks. 
Yamamoto et. al.~\cite{DaisukeYamamoto20242023XBL0180} proposed a data transfer optimization algorithm to maximize utility for both autonomous mobilities and users.
Most of these studies focus on optimizing the throughput of systems where each data generation location is visited only once for data pickup.
In contrast, our work focuses on the freshness of data, specifically the AoI, for recurring data generation at the data nodes.

\subsection{Path Planning Problems}

Several path planning problems closely related to our work include the \emph{Traveling Salesman Problem} (TSP)~\cite{TSP}, the \emph{Traveling Repairman Problem} (TRP)~\cite{TRP}, and the \emph{Vehicle Routing Problem} (VRP)~\cite{VRP'59}.
Below, we provide a brief introduction to these problems and highlight the distinctions between them and our specific focus.

%%% TSP
Given a set of nodes and the distances between each pair of nodes, the \emph{Traveling Salesman Problem} (TSP)~\cite{TSP} seeks to find the shortest route that visits each node exactly once and returns to the starting node.
Known as an NP-hard problem, exact solutions are computationally intensive; for instance, the Held-Karp algorithm~\cite{bellman1962dynamic} can solve the TSP in $O(n^2 2^n)$ time using dynamic programming.
To approximate solutions efficiently, various heuristics and approximation algorithms have been developed, such as the Nearest Neighbor algorithm, the Christofides algorithm~\cite{christofides2022worst}, and the Lin–Kernighan heuristic~\cite{LKHeuristic}.
Our problem differs from the TSP, as demonstrated in Section~\ref{sec:edge_enforcement}, where we show that a shorter route does not necessarily result in a smaller AoI.

%%% TRP
The \emph{Traveling Repairman Problem} (TRP)~\cite{TRP}, also known as the delivery man problem or the minimum latency problem, is an extension of the TSP in combinatorial optimization.
In the TRP, given a set of customers located at different points, the objective is to find a route for a repairman that minimizes the sum of delays for reaching each customer.
This problem is NP-hard for general graphs.
The TRP differs from our problem as the TRP aims to minimize the time it takes to reach each customer, optimizing for delay.

%%% VRP
The vehicle routing problem (VRP)~\cite{VRP'59} is a NP-Hard combinatorial optimization problem which generalizes the TSP.
It involves determining the optimal set of routes for a fleet of vehicles to serve a given set of customers.
There are many variants of the VRP, such as the capacitated VRP, the VRP with time windows, and the VRP with pickup and delivery.
Our problem can be viewed as a variant of the VRP, where there is a single vehicle (the drone) tasked with collecting data recurrently generated at the nodes.
However, our problem diverges from traditional VRP variants because it involves repetitive data generation at each node.
This repetitive data generation introduces a dynamic element that is not typically considered in standard VRP formulations, where nodes are visited only once.

\section{System Model}

We first present the system model of a patrolling data collection drone framework in Section~\ref{sec:model_component}.
In Section~\ref{sec:AoI}, we define \emph{Age-of-Information} for individual data packets and introduce the concept of \emph{Maximum Age-of-Information} (MAI) for both data nodes and the route.

\subsection{System Components}
\label{sec:model_component}

%%% framework
In a patrolling data collection drone framework, data nodes deployed in a confined area \emph{recurrently} generate data packets that must be collected and transported to a server for aggregation and analysis.
These data nodes generally remain stationary, but can be relocated if necessary.
A data collection drone patrols the area along a predetermined \textit{patrolling route}, collecting data packets from these nodes and transporting them back to the server.
Due to the limited transmission range of high-frequency wireless communication, we assume that data exchange can only occur when the drone is positioned directly above a node.

%%% complete graph
Given a server node and a set of $N$ stationary data nodes, the system can be modeled as a complete graph $G = (V,E)$, where $V$ represents the set of nodes and $E$ as the set of undirected edges, each weighted by the \textit{minimum direct travel times} between every pair of nodes.
The \emph{minimum direct travel time} $t_{i,j}$ between nodes $i$ and $j$ is computed based on the distance between the nodes and the drone's maximum possible speed.
All edges are assumed to be nonnegative, symmetric, i.e., $t_{i,j}  = t_{j,i}$, and follow the triangle inequality, i.e., $t_{i,j} \leq t_{i,k} + t_{k,j}$ for all $i, j, k \in V$.

%%% patrolling route
A patrolling route $R$ is a sequence of nodes forming a Hamiltonian circuit in the graph $G$ that starts and ends at the server node $v_0$.
Without loss of generality, we denote the data nodes as $v_{R,1}, v_{R,2}, \ldots, v_{R,N}$ in the order they are visited in the patrolling route $R$ in the remaining sections.
In this study, we specifically address scenarios where each data node is visited \textbf{exactly once} within a route.

Given a patrolling route $R$, the travel time between two nodes $i$ and $j$ \textit{along the route} is denoted as $T_{R}(i, j)$.
Due to the high data rate of high-frequency wireless communication technology, e.g., up to $100~Gbps$ for mmWave IEEE 802.11.ay~\cite{802.11ay}, the data exchange time between the drone and a data node can be considered negligible compared to the travel time.
This enables the drone to collect all accumulated data packets at a data node upon arrival and immediately proceed to the next node without pausing.
Thus, we have:
\begin{equation}
	T_{R}(i, j) = \sum_{k=i}^{j-1} t_{k, k+1},
\end{equation}
where $t_{k, k+1}$ is the minimum direct travel time between node $k$ and node $k+1$.
The \emph{round trip time} of a route $R$, denote as $T_{R}$, is the total travel time from the server node to the last visited data node and back to the server.
That is,
\begin{equation}
    T_R = \sum_{i=0}^{N-1} t_{i, i+1} + t_{N, 0}
\end{equation}

\subsection{Age-of-Information}\label{sec:AoI}

The \emph{Age-of-Information} (AoI) of a data packet is defined as the duration between the generation of the data packet and its processing at the server.
AoI reflects the freshness of the data when being analyzed at the server.
For simplicity, we assume that data packets are processed immediately upon their arrival at the server.
Thus, the AoI of a data packet consists of the time it spends waiting to be collected by the drone and the duration it takes for the drone to traverse back to the server.
For data packets generated at data node $v_{R,i}$, we formulate the AoI as follows:
\begin{equation}\label{eq:AoI}
	\text{AoI}_{R}(i) = (T_{R}(i,0) - \Delta) + T_{R}(0,i) + T_{R}(i,0),
\end{equation}
where $\Delta$ is the time between the last visit of the drone at $v_{R,i}$ and the generation of the data packet.

%%% MAI
The \emph{Maximum Age-of-Information} (MAI) of a patrolling route $R$ is defined as the largest possible AoI among all data packets arriving at the server.
For data packets generated at data node $v_{R,i}$, the packet with the largest AoI is generated \textbf{immediately after} the drone departs from $v_{R,i}$ (i.e., $\Delta \rightarrow 0$) in the worst-case scenario.
Consequently, this packet must wait for the drone to complete its patrol, return to $v_{R,i}$, and then complete the remaining route back to the server.
The MAI for data node $v_{R,i}$, representing the largest possible AoI of data packets generated at $v_{R,i}$, can be formulated as follows:
\begin{equation}\label{eq:MAI}
    \begin{split}
        \text{MAI}_{R}(i) &= T_{R}(i,0) + T_{R}(0,i) + T_{R}(i,0) \\
		&= T_{R} + T_{R}(i, 0)
	\end{split}
\end{equation}
Since the MAI of a patrolling route $R$ is the largest possible AoI among all data nodes, we have:
\begin{equation}
	\text{MAI}_{R} = \max_{i}(\text{MAI}_R(i))
\end{equation}

\section{AoI-Aware Route Planning}

We explore strategies for designing a patrolling route with optimal MAI in this section.
We begin by defining the \emph{AoI-aware route planning problem} in Section~\ref{sec:problem_definition} and proving that its decision version is NP-Complete.
Given this complexity, we present two approaches with approximation guarantees rather than solving the problem optimally.
Our initial approach, which focuses on minimizing the shortest round trip time, is presented in Section~\ref{sec:shortest_route}.
This approach achieves an approximation ratio of $1.5$ with respect to MAI.
In Section~\ref{sec:edge_enforcement}, we first show that a route with a shorter round trip time does not necessarily result in a lower MAI.
We then propose an alternative approach, \emph{Edge Enforcement}, which also provides an approximation ratio of $1.5$ for MAI.

\subsection{Problem Definition}\label{sec:problem_definition}
%%% problem definition
Ensuring the freshness of data is critical for the effectiveness of patrolling drone systems.
Given a complete graph $G = (V, E)$, where $V$ is the set of nodes and $E$ is the set of undirected edges with weights corresponding to the minimum direct travel times between each pair of nodes, satisfying the triangle inequality, our goal is to design a patrolling route $R$ for the drone that minimizes the MAI.
Formally, our objective function is:
\begin{equation}
    \argmin_{R} (\text{MAI}_{R}) = \argmin_{R} (\max_{i} (\text{MAI}_{R}(i)))
\end{equation}
Recall that a patrolling route $R$ is a sequence of nodes forming a Hamiltonian circuit in $G$ that starts and ends at the server node $v_0$.
By applying a patrolling route with an optimal MAI, we can ensure that data carried by the drone arrives at the server node with a certain level of freshness.

We have the following lemma for the MAI of a route:
\begin{lemma}\label{lmm:first_node}
	Given a patrolling route $R$, the MAI of the route comes from the data packets generated at the first visited data node $v_{R,1}$.
	That is, 
	\begin{equation}
		\text{MAI}_{R} = \max_{i} (\text{MAI}_{R}(i)) = \text{MAI}_{R}(1) \nonumber
	\end{equation}	
\end{lemma}
\begin{proof}
	According to Equation~\eqref{eq:MAI}, the MAI of a data node is the sum of the round trip time $T_R$ and the travel time from the node back to the server.
	Since $T_R$ remains the same for all data nodes in a given route, and $T_{R}(1, 0)$ is larger than $T_{R}(i, 0)$ for every $i > 1$, the MAI of the first visited data node must be the largest among all data nodes, which by definition is the MAI of the route.
\end{proof}

With the above lemma, we prove that the decision version of the \emph{AoI-aware route planning problem} is NP-Complete in the following theorem.
Specifically, the decision version of the problem is to \textit{determine whether there exists a patrolling route with an MAI no more than a given threshold}.

\begin{thm}\label{thm:MAI_NP_C}
	Determine whether there exists a patrolling route with an MAI no more than a given threshold is NP-Complete.
\end{thm}
\begin{proof}
       Validating whether a patrolling route with an MAI no more than a given threshold can be done in polynomial time. Hence, the problem is in NP.
	We prove the hardness by reducing from the \emph{Hamiltonian Path problem}, a well-known NP-Complete problem~\cite{Karp1972}, to this decision problem.
	Consider a non-complete undirected graph $G^* = (V^*, E^*)$ with $N$ vertices as an instance of the Hamiltonian Path problem.
	If a Hamiltonian path exists in $G^*$, then there is a path that visits every vertex in $V^*$ \emph{exactly once}.

	We can construct an instance $G = (V, E)$ of the \emph{AoI-aware route planning problem} from $G^*$ in polynomial time as follows:
        \begin{itemize}
        \item $V = V^* \cup \{u\}$, where $u$ is the server node.
        \item For any two nodes $i$ and $j$ in $V^*$, if $(i, j)$ is
          in $E^*$, then $t_{i,j}$ is set to $1$; otherwise, $t_{i,j}$
          is set to $2$. For the vertex $u$, let $t_{u, i}$ be $1$
          for every $i \in V^*$. These edges form the set $E$. We
          note that these edges are nonnegative, symmetric, and follow the triangle inequality.
        \end{itemize}        
	We then prove that there exists a patrolling route with an MAI no more than $2N+1$ in $G$ if and only if there exists a Hamiltonian path in $G^*$.	
	\begin{itemize}
		\item If a Hamiltonian path exists in $G^*$, we can construct a patrolling route in $G$ that starts from $u$ and visits each node in $V$ exactly once, and return to $u$.
		Since the weight of the corresponding Hamiltonian path in $G$ is $N-1$, according to Lemma~\ref{lmm:first_node}, the MAI of the route is $(1 + N-1 + 1) + (N-1 + 1) = 2 \cdot N + 1$.
        \item If there exists a patrolling route with an MAI no more than $2N+1$ in $G$, this route must be a patrolling route of $G$ in which each edge has a cost of $1$.
		By removing the two edges connected to $u$ in this route, the remaining is a path which visits every vertex \emph{exactly once}.
		This path in the corresponding graph $G^*$ is a Hamiltonian path.
	\end{itemize}	
\end{proof}

\subsection{Shortest Round Trip Time}\label{sec:shortest_route}

%%% opening
According to Theorem~\ref{thm:MAI_NP_C}, the decision version of the \emph{AoI-aware route planning problem} is NP-Complete.
To address this challenge, we propose two approaches with approximation guarantees to address the \emph{AoI-aware route planning problem}.
Our first approach, \emph{shortest round trip time} (SRTT), aims to minimize the round trip time of the route.
The intuition behind this approach is that the MAI of a route is largely influenced by the round trip time. 
This can be observed by rewriting Equation (4) into the following form:
\begin{equation}
    \begin{split}
        \text{MAI}_R(i)	&= T_R + T_{R}(i,0) \\
                        &= T_R + (T_R - T_{R}(0,i)) \\						
						&= 2 \cdot T_R - T_{R}(0,i),
    \end{split}
    \label{eq:MAI_2}
\end{equation}
where $T_{R}(0,i)$ is the travel time from the server to the $i$-th visited data node.
Thus, reducing the round trip time should ideally result in a smaller MAI.

Finding a route with the shortest round trip time is equivalent to solving the Traveling Salesman Problem (TSP)~\cite{lawler1985travelling}, which is a well-known NP-hard problem.
Various heuristics and approximation algorithms exist for solving the TSP.
In this work, we employ the Christofides algorithm~\cite{christofides2022worst} to generate a route for the drone.
The algorithm guarantees a solution within $1.5$ times the optimal route length.

%%% Our approach
Our SRTT approach, based on the Christofides algorithm, works as follows.
Given a complete graph $G = (V, E)$, the algorithm first constructs a minimum spanning tree (MST) that connects all nodes with the minimum total edge weight.
Recall that $V$ is the set of ground nodes, and $E$ is the set of edges with the minimum direct travel times between each pair of nodes as weights.
Two common algorithms for constructing an MST are \emph{Prim's algorithm}~\cite{6773228} and \emph{Kruskal's algorithm}~\cite{kruskal1956shortest}.
Our approach employs Prim's algorithm to build the MST, starting from the server node $v_0$ and iteratively adding the shortest edge that connects the current tree to an unvisited node.

Once the MST is established, the algorithm identifies the odd-degree vertices within the MST and finds the \emph{minimum-weight perfect matching} for these vertices.
By combining the MST with this matching, the algorithm constructs an Eulerian circuit that visits each edge in the combined graph.
This circuit is then converted into a TSP tour by removing repeated nodes, ensuring that each node is visited exactly once.

The following lemma provides an important subroutine for the analysis of our algorithms in this paper.

\begin{lemma}\label{lmm:perfect-matching}
  Given a complete graph $G = (V, E)$, the cost of the minimum-weight
  perfect matching for the odd-degree vertices for the nodes of the
  spanning tree in Christofides algorithm is at most $0.5$ of the
  optimal TSP tour.
\end{lemma}
\begin{proof}
  This comes from Lemma 1 in~\cite{christofides2022worst}.
\end{proof}

The final step involves converting the undirected TSP tour into a directed patrolling route.
This is achieved by selecting one of the two nodes connected to the server node $v_0$ in the TSP tour as the first visited node.
According to Equation~\eqref{eq:MAI_2}, the MAI is twice the round trip time minus the travel time from the server to the first visited data node.
To minimize the MAI, we choose the node with the longer direct travel time from the server node as the first node to be visited.
This selection process ensures that the resulting directed patrolling route achieves a smaller MAI.

\medskip
\noindent \textbf{\textit{Approximation Ratio}}: Since our approach uses the Christofides algorithm, which provides a $1.5$-approximation for the TSP, the resulting patrolling route may not be optimal in terms of both round trip time and MAI.
Therefore, we derive the approximation ratio of \emph{SRTT} in terms of the MAI with the following theorem:

\begin{thm}\label{thm:SRTT_approx}
	The Shortest Round Trip Time (SRTT) approach provides an approximation ratio of $1.5$ in terms of MAI.
\end{thm}
\begin{proof}
	Given a complete graph $G = (V, E)$, where $V$ is the set of ground nodes and $E$ is the set of edges with the minimum direct travel times between each pair of nodes as weights, we define $W_{MST}$ as the weight sum of the minimum spanning tree of $G$, and $W_{TSP}$ as the weight sum of the optimal TSP tour in $G$.
	By definition, $W_{MST} \leq W_{TSP}$, as the removal of an edge of a TSP tour yields a spanning tree.

	Recall that the MAI is the round trip time of the route plus the travel time from the first visited data node back to the server. Therefore, the optimal MAI must be at least $W_{TSP} + W_{MST}$.
        
	In \emph{SRTT}, we employ the Christofides algorithm, which combines the minimum spanning tree (with a cost of $W_{MST}$) and the minimum-weight perfect matching (with a cost upper bounded by $0.5 W_{TSP}$ due to Lemma~\ref{lmm:perfect-matching}) to create an Eulerian circuit.
	A patrolling route $R$ is derived from the Eulerian circuit by taking shortcuts.
	Since edges in $G$ follow the triangle inequality, the round trip time of the route $T_R$ (i.e., the weight sum of the route $R$) is guaranteed to be upper bounded by the cost of the minimum spanning tree and the minimum-weight perfect matching, i.e., $W_{MST} + 0.5 \cdot W_{TSP}$.
	
	Nevertheless, since the MAI is at most twice the round trip time according to Equation~\eqref{eq:MAI_2}, the MAI of the route $R$ generated by the SRTT approach is at most $2 \cdot (W_{MST} + 0.5 \cdot W_{TSP}) = 2 \cdot W_{MST} + W_{TSP}$.
	Therefore, \emph{SRTT} provides an approximation ratio of
	\begin{equation}
		\frac{2 \cdot W_{MST} + W_{TSP}}{W_{MST} + W_{TSP}} = 1 + \frac{W_{MST}}{W_{MST} + W_{TSP}}\leq 1.5,
    \end{equation}
    where the inequality is due to $W_{MST} \leq W_{TSP}$.	
\end{proof}

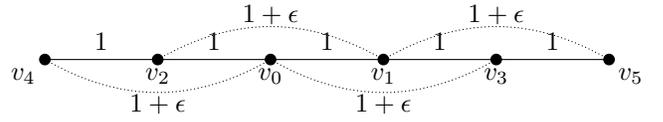
\begin{figure}
	\centering
	% \documentclass{standalone}
% \usepackage{tikz}
% \usetikzlibrary{graphdrawing, graphs, shapes.geometric, positioning}
% \usepackage{amsmath,amsfonts,amssymb}

% \begin{document}
\begin{tikzpicture}
    % Define nodes
    \coordinate (v4) at (0,0);
    \coordinate (v2) at (1.5,0);
    \coordinate (v0) at (3,0);
    \coordinate (v1) at (4.5,0);
    \coordinate (v3) at (6,0);
    \coordinate (v5) at (7.5,0);
    % \coordinate (v5) at (4,2);

    \filldraw (v4) circle (2pt) node[anchor=north east] {$v_4$};
    \filldraw (v2) circle (2pt) node[anchor=north] {$v_2$};
    \filldraw (v0) circle (2pt) node[anchor=north] {$v_0$};
    \filldraw (v1) circle (2pt) node[anchor=north] {$v_1$};
    \filldraw (v3) circle (2pt) node[anchor=north] {$v_3$};
    \filldraw (v5) circle (2pt) node[anchor=north west] {$v_5$};    
    
    % Draw edges
    \draw (v4) -- (v2) node[midway, above] {$1$};
    \draw (v2) -- (v0) node[midway, above] {$1$};
    \draw (v0) -- (v1) node[midway, above] {$1$};
    \draw (v1) -- (v3) node[midway, above] {$1$};
    \draw (v3) -- (v5) node[midway, above] {$1$};
        
    % \draw [densely dotted] (v1) -- (v5) node[midway, left] {$L$};
    % \draw [densely dotted] (v2) -- (v5) node[midway, left] {$L$};
    % \draw [densely dotted] (v3) -- (v5) node[midway, left] {$L$};
    % \draw (v4) -- (v5) node[midway, left] {$L$};    
    
    % Draw remaining edges based on shortest paths    
    \draw [densely dotted] (v4) to [bend right=30] (v0);
    \node[] at (1.5,-0.6) {$1+\epsilon$};
    \draw [densely dotted] (v2) to [bend left=30] (v1);
    \node[] at (3,0.6) {$1+\epsilon$};
    \draw [densely dotted] (v0) to [bend right=30] (v3);
    \node[] at (4.5,-0.6) {$1+\epsilon$};
    \draw [densely dotted] (v1) to [bend left=30] (v5);
    \node[] at (6,0.6) {$1+\epsilon$};
    % \draw [densely dotted] (v0) -- (v5) node[midway, left] {$L+\epsilon$};
\end{tikzpicture}

% \end{document}		
	\caption{A subfigure of the graph $G$ used for the tightness analysis of the proposed approaches. The minimum spanning tree of $G$ is shown in solid line. In this example, $t_{0,1} = t_{0,2} = t_{1,2} = 1$ and $t_{1,3} = t_{2,3} = 100$}
	\label{fig:tight_example}
	\vspace{-1em}
\end{figure}

This approximation ratio is tight as can be seen by the following example.
Let $v_0$ be the server node, $v_1, v_2, \ldots, v_{N}$ be the $N=2M+1$ data nodes, and $\epsilon$ be a positive real number close to zero.
For the input graph, set $t_{0,1} = t_{i,i+2} = 1, \forall i \in \{0, 1, \ldots, N-2\}$, and $t_{0,3} = t_{1,2} = t_{i,i+4} = 1+\epsilon, \forall i \in \{0, 1, \ldots, N-4\}$.
All remaining edges in the complete graph have distances given by the shortest paths within this subgraph.
Figure~\ref{fig:tight_example} illustrates this subgraph with $N=5$.

Given the above graph, \emph{SRTT} generates a route as follows:
The minimum spanning tree of $G$ includes all edges with weight $1$, resulting in $W_{MST} = N-1$.
With only two odd-degree vertices in the MST, the minimum-weight perfect matching consists of one edge, $(v_{N-1}, v_{N})$, with weight at most $(N-1)(1+\epsilon)/2 + 1$.
The union of the tree and the matching forms a cycle, with no possible shortcuts.
Thus, the MAI of this route is at most $(N-2)(1+\epsilon)+2N-1$.

Consider an alternative route which visits $v_0, v_{1}, v_{5}, \ldots, \allowbreak  v_2, v_0$.
The round trip time of this route is at least $1+(N-1)(1+\epsilon)+1$, and resulting in an MAI of $2(N-1)(1+\epsilon) +3$.
As $\epsilon \rightarrow 0$, the ratio of the MAI derived from \emph{SRTT} to that of the alternative route is at least $\lim_{\epsilon \rightarrow 0}\frac{(N-2)(1+\epsilon)+2N-1}{2(N-1)(1+\epsilon) +3} = 3/2$ as $N \rightarrow \infty$.
Combined with Theorem~\ref{thm:SRTT_approx}, we conclude that the approximation ratio of $1.5$ for the SRTT approach is tight.

\subsection{Edge Enforcement}\label{sec:edge_enforcement}

%%% counter example
\begin{figure}[tp]
	\centering
	\begin{subfigure}{.49\linewidth}
		\centering
		\begin{tikzpicture}
			% \documentclass[tikz]{standalone}

% \usepackage{pgfplots}
% \usetikzlibrary{shapes.symbols}
% \usepackage{amsmath,amsfonts,amssymb}

% \begin{document}

% \begin{tikzpicture}
    \coordinate (A) at (0,0);
    \coordinate (B) at (0.5,0.13);
    \coordinate (C) at (0.13,0.5);
    \coordinate (D) at (2,2); % Adjusted scale for better visualization
    
    % Nodes
    \filldraw (A) circle (2pt) node[anchor=north east] {$v_0$};
    \filldraw (B) circle (2pt) node[anchor=north west] {$v_1$};
    \filldraw (C) circle (2pt) node[anchor=south east] {$v_2$};
    \filldraw (D) circle (2pt) node[anchor=south east] {$v_3$};
    
    % Path
%     \draw[->, thick] (A) -- (B);
%     \draw[->, thick] (B) -- (D);
%     \draw[->, thick] (D) -- (C);
%     \draw[->, thick] (C) -- (A);
% \end{tikzpicture}

% \end{document}
			% Path
			\draw[->, thick] (A) -- (B);
			\draw[->, thick] (B) -- (D);
			\draw[->, thick] (D) -- (C);
			\draw[->, thick] (C) -- (A);
			\node[] at (0.25, -0.25) {1};
			\node[] at (1.5, 0.75) {100};
			\node[] at (0.5, 1.25) {100};
			\node[] at (-0.25, 0.25) {1};
		\end{tikzpicture}
		\caption{Shortest route}
	\end{subfigure}%
	\begin{subfigure}{.49\linewidth}
		\centering
		\begin{tikzpicture}
			% \documentclass[tikz]{standalone}

% \usepackage{pgfplots}
% \usetikzlibrary{shapes.symbols}
% \usepackage{amsmath,amsfonts,amssymb}

% \begin{document}

% \begin{tikzpicture}
    \coordinate (A) at (0,0);
    \coordinate (B) at (0.5,0.13);
    \coordinate (C) at (0.13,0.5);
    \coordinate (D) at (2,2); % Adjusted scale for better visualization
    
    % Nodes
    \filldraw (A) circle (2pt) node[anchor=north east] {$v_0$};
    \filldraw (B) circle (2pt) node[anchor=north west] {$v_1$};
    \filldraw (C) circle (2pt) node[anchor=south east] {$v_2$};
    \filldraw (D) circle (2pt) node[anchor=south east] {$v_3$};
    
    % Path
%     \draw[->, thick] (A) -- (B);
%     \draw[->, thick] (B) -- (D);
%     \draw[->, thick] (D) -- (C);
%     \draw[->, thick] (C) -- (A);
% \end{tikzpicture}

% \end{document}
			% Path
			\draw[->, thick] (A) -- (D);
			\draw[->, thick] (D) -- (B);
			\draw[->, thick] (B) -- (C);
			\draw[->, thick] (C) -- (A);
			\node[] at (0.33, 0.55) {1};
			\node[] at (1.5, 0.75) {100};
			\node[] at (0.75, 1.25) {$t_{0,3}$};
			\node[] at (-0.25, 0.25) {1};	
		\end{tikzpicture}
		\caption{Route with the minimum MAI}
	\end{subfigure}	
	\caption{An example where the shortest route has a larger MAI compared to a route with a longer round trip time.}
	\label{fig:counter_example}
	\vspace{-1em}
\end{figure}
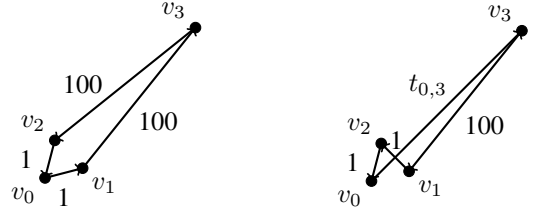

%%% opening
In our first approach, we attempt to minimize the round trip time of the route to reduce the MAI.%, as the MAI is typically dominated by the round trip time.
However, a shorter round trip time does not guarantee a smaller MAI.
Figure~\ref{fig:counter_example} illustrates an example where the shortest route has a larger MAI compared to a route with a longer round trip time.

%%% example
Consider a scenario with a server node $v_0$ and three data nodes $v_1, v_2, v_3$.
The minimum direct travel times between nodes are $t_{0,1} = t_{0,2} = t_{1,2} = 1$, $t_{1,3} = t_{2,3} = 100$, and $t_{0,3} > t_{1,3}$.
By the triangle inequality, we know that $t_{0,3} < 101$.
The shortest route $R^{*}$ is $v_0 \rightarrow v_1 \rightarrow v_3 \rightarrow v_2 \rightarrow v_0$ with a round trip time of $202$ and an MAI of $403$.

Now consider another route $R: v_0 \rightarrow v_3 \rightarrow v_1 \rightarrow v_2 \rightarrow v_0$.
Since $t_{0,3} > t_{1,3} =100$, the round trip time $t_{0,3} + 102 > 202$ is longer than that of the shortest route $R^{*}$.
However, the MAI of this route is $2 \cdot (t_{0,3} + 102) - t_{0,3} = t_{0,3} + 204 < 305$, which is significantly smaller than the MAI of $R^{*}$.

This example demonstrates that a shorter route does not necessarily lead to a smaller MAI.
The MAI is affected not only by the round trip time of the route but also by the travel time between the server and the first visited data node.

To investigate the impact of the first edge on the MAI, we can rewrite the MAI of a route using Lemma~\ref{lmm:first_node} and Equation~\eqref{eq:MAI}.
The MAI can be expressed as the travel time from the server to the first visited node plus twice the travel time from the first visited node back to the server, i.e., 
\begin{equation}
	\text{MAI}_{R}(1) = T_{R} + T_{R}(1,0) = T_{R}(0,1) + 2 \cdot T_{R}(1,0). \nonumber
\end{equation} 
Since $T_{R}(0,1)$ equals to the minimum direct travel time from the server node $v_0$ to the first visited node $v_{R,1}$, which is a constant in $G$, the minimum MAI of route $R$ with a specified first edge can be achieved by minimizing $T_{R}(1,0)$.

\begin{thm}\label{thm:enforcenode}
	Suppose that $R_i$ is a patrolling route of $G$ with a specified edge $(v_0, v_i)$.
	If the remaining part of $R_i$ is the shortest Hamiltonian path from $v_i$ to $v_0$, then minimum MAI can be achieved by taking
	\[
	\min_{i} \left(t_{0,i} + 2 \cdot T_{R_i}(1,0)\right).
	\]
\end{thm}
\begin{proof}
	We prove this theorem by contradiction.
	Assume that there exists a patrolling route $R_i$ with a specified edge $(v_0, v_i)$ such that the minimum MAI of $R_i$ is \textbf{not} achieved by the shortest Hamiltonian path from $v_i$ to $v_0$.
	We denote the shortest Hamiltonian path from $v_i$ to $v_0$ as $R^*_i$, where $T_{R^*_i}(1,0) \leq T_{R_i}(1,0)$.
	According to Equation~\eqref{eq:MAI}, the minimum MAI of $R_i$ equals to $t_{0,i} + 2 \cdot T_{R_i}(1,0) \geq t_{0,i} + 2 \cdot T_{R^{*}_{i}}(1,0)$, which contradicts our assumption.		
	Since the minimum MAI of $R_i$ is achieved by the shortest Hamiltonian path from $v_i$ to $v_0$, iterating through all possible nodes ensures that the minimum MAI can be found.
\end{proof}

%%% DP for the optimal solution
Based on Theorem~\ref{thm:enforcenode}, we can use the following dynamic programming approach to find the optimal route with the minimum MAI.
The approach iteratively enforces each data node to be the first visited node in the route.
For each enforced node, the Held-Karp algorithm~\cite{bellman1962dynamic} is applied to solve the Shortest Hamiltonian Path problem with specified starting and ending nodes.
The Held-Karp algorithm computes the shortest distance from the starting node to all other nodes while passing through every other node as intermediate results.
By setting $v_i$ as the starting node and $v_0$ as the ending node, we can determine the length of the shortest Hamiltonian path between these two nodes.

Although the minimum MAI can be achieved using this dynamic programming approach, its time complexity grows exponentially with the number of nodes.
The Held-Karp algorithm has a time complexity of $O(2^N \cdot N^2)$, and it is triggered once for each of the enforced nodes.
Therefore, the overall time complexity of the dynamic programming approach is $O(2^N \cdot N^3)$, making it infeasible for large $N$.

%%% proposed algorithm
\noindent \textbf{Our Approach}:
Instead of finding the optimal route with the minimum MAI in exponential time, we propose an approximation approach called \emph{Edge Enforcement} (Enforced).
The algorithm starts by enforcing an arbitrary edge connected to the server node and follows a sequence of steps to construct a route.
First, a spanning tree that includes the enforced edge is built.
This construction begins with the server node and the enforced edge, then incrementally incorporates the remaining nodes using Prim's algorithm.
Notably, the resulting spanning tree with the enforced edge will have a total weight as least as large as that of the minimum spanning tree.

Once the spanning tree with the enforced edge is established, the Christofides algorithm is applied to generate a patrolling route aimed at minimizing the round trip time. 
The algorithm proceeds by constructing a minimum-weight perfect matching among the odd-degree vertices within the spanning tree, combining the perfect matching with the spanning tree to form an Eulerian circuit, and converting the Eulerian circuit into a TSP tour.
Throughout this conversion, the enforced edge remains as the initial edge.

After generating a TSP tour, the algorithm forms two routes: one starting with the enforced edge and the other in the opposite direction.
The MAI for both routes is calculated, and compared with that of the current best route.
The route with the smallest MAI is kept as the new best route.
This process is repeated for each edge connected to the server node until all have been enforced at least once.

The time complexity of \emph{Enforced} is dominated by the Christofides algorithm, which is $O(|E| \log |V|)$ for constructing the minimum spanning tree and $O(|V|^3)$ for the minimum-weight perfect matching.
Since the algorithm is applied once for each edge connected to the server node, with at most $N$ such edges, the overall time complexity of \emph{Enforced} is $O(N^4)$.

\medskip
\noindent \textbf{\textit{Approximation Ratio}}:
We now derive the approximation ratio of \emph{Enforced} in terms of the MAI.

\begin{thm}
	\emph{Enforced} has an approximation ratio of $1.5$.
\end{thm}
\begin{proof}
	We first establish that the MAI of the route generated by \emph{Enforced} is no greater than that of the route produced by \emph{SRTT}.
	Given a complete graph $G$, \emph{SRTT} finds a patrolling route by first forming the minimum spanning tree of $G$.
	When \emph{Enforced} enforces the shortest edge connected to the server node, it constructs the same spanning tree as \emph{SRTT}.
	Since the remaining steps in the route construction are identical for both approaches, they result in the same patrolling route with the same MAI.
	However, \emph{Enforced} explores additional routes by enforcing other edges connected to the server node and selects the route with the minimum MAI among them.
	As a result, the MAI of the route generated by \emph{Enforced} can only be smaller than or equal to that of the route produced by \emph{SRTT}.	

	According to Theorem~\ref{thm:SRTT_approx}, the MAI of the route generated by \emph{SRTT} is at most $1.5$ times the optimal MAI. 
	Since \emph{Enforced} guarantees a MAI no greater than that of \emph{SRTT}, the approximation ratio of \emph{Enforced} is at most $1.5$.
\end{proof}

%%% Tightness
Although \emph{Enforced} offers a $1.5$-approximation ratio, the tightness of this approximation ratio remains open.
The input instance used in the tightness analysis of \emph{SRTT} (Figure~\ref{fig:tight_example}) shows that the approximation ratio of \emph{Enforced} is at least $1.375$.
For this instance, \emph{Enforced} yields a route with the enforced edge $(v_0, v_{N})$, resulting in a round trip time at least $(N-1)(1+\epsilon)/2 + (N-1)$ and an MAI at least $3(N-1)(1+\epsilon)/4 + 2 \cdot (N-1)$.
As $\epsilon \rightarrow 0$, the ratio of the MAI derived from \emph{Enforced} to that of the optimal route becomes $\lim_{\epsilon \rightarrow 0}\frac{3(N-1)(1+\epsilon)/4 + 2 \cdot (N-1)}{2(N-1)(1+\epsilon) +3} = \frac{11(N-1)/4}{2N+1}$, which approaches $1.375$ as $N \rightarrow \infty$.

\section{Evaluation}\label{sec:eval}

We evaluate the performance of our proposed approaches by measuring the \emph{Maximum Age-of-Information} (MAI) of the routes they generate and the computation time required to construct these routes.
The evaluations are conducted on a set of synthetic scenarios with varying numbers of nodes and different node distributions.
Section~\ref{sec:eval_Settings} describes the setup and procedure for generating these scenarios.
In Section~\ref{sec:eval_MAI}, we analyze the MAI of routes generated by different approaches and introduce a hybrid method.
Finally, in Section~\ref{sec:eval_Exec} we compare the computation time of each approach.

\subsection{Settings}\label{sec:eval_Settings}

The evaluations are conducted on a machine with an Intel Core i7-8565U CPU@1.80~GHz and 16~GB of RAM. 
The approaches are implemented in C++ and compiled with g++ 9.4.0.
We synthesize a series of scenarios, each consisting of the coordinates of a server node and a set of stationary data nodes.
These scenarios serve as inputs for the approaches to generate the patrolling routes.
We measure the MAI of the generated routes and the computation time required to generate these routes.

%%% data sets
\noindent \textbf{Data Sets}:
The synthetic scenarios are generated based on the following settings:
We consider two sets of scenarios with varying numbers of nodes: \emph{8-node} and \emph{20-node}.
A \emph{8-node} scenario consists of $8$ data nodes within a $1~km \times 1~km$ area, while a \emph{20-node} scenario includes $20$ data nodes within a $8~km \times 8~km$ area.
In each scenario, the coordinates of the server node are located at the center of the area, while the data nodes are distributed within the given area.

We employ three types of distributions for the data nodes: \emph{grid}, \emph{cluster}, and \emph{outlier}.
For both the \emph{grid} and \emph{cluster} distributions, the area is divided into a grid of $4 \times 4$ cells.
In the \emph{grid} distribution, each cell contains a similar number of data nodes.
In the \emph{cluster} distribution, only \textbf{one} cell contains data nodes for the \emph{8-node} scenario, and \textbf{four} cells contain data nodes for the \emph{20-node} scenario.
For the \emph{outlier} distribution, the area is divided into a grid of $2 \times 2$ cells, with only \textbf{two} cells containing data nodes: one cell has a single data node as an outlier, and the other cell has the remaining data nodes.

Given the number of nodes and the node distribution of a scenario, cells are randomly selected to place the data nodes.
Within each cell, the coordinates of the data nodes are randomly generated.
For each combination of the number of nodes and node distribution, we generate $100$ scenarios as our input data set.

%%% travel time
For simplicity, we assume that the drone maintains the maximum speed while traveling between nodes.
The minimum direct travel time between two nodes is calculated using the Euclidean distance between them divided by the maximum speed of the drone.
Given that most commonly used drones can achieve a speed of $70~km/h$ (approximately $19.44~m/s$), we set the maximum speed of the drone to $20~m/s$.

%%% evaluation metrics
\noindent \textbf{Comparing Approaches}:
We compare the following route generation approaches:
\begin{itemize}
	\item \emph{Shortest Round Trip Time} (SRTT): Our first approach, presented in Section~\ref{sec:shortest_route}, based on the Christofides algorithm. \emph{SRTT} aims to minimize the MAI by reducing the total round trip time of the route.
	\item \emph{Edge Enforcement} (Enforced): Our second approach, presented in Section~\ref{sec:edge_enforcement}. \emph{Enforced} iteratively enforces each edge connected to the server node as part of the route and selects the one with the smallest MAI among all generated routes.
	\item \emph{Nearest Neighbor} (Greedy): A simple heuristic that constructs a route by selecting the nearest unvisited node as the next node to visit.
	\item \emph{Dynamic Programming} (DP): An approach introduced in Section~\ref{sec:edge_enforcement} that computes the optimal route with the minimum MAI.
	\item \emph{Lin-Kernighan-Helsgaun TSP Solver} (LKH): A state-of-the-art TSP solver that aims to find the optimal route with the minimum round trip time. For our evaluation, we utilize LKH-3~\cite{LKH-3}, an implementation of the Lin-Kernighan-Helsgaun algorithm~\cite{LKHeuristic}.
\end{itemize}

\subsection{Maximum Age-of-Information (MAI) of Routes}\label{sec:eval_MAI}

\begin{figure*}[tb]
  \centering
	\begin{subfigure}{0.32\textwidth}
		\includegraphics[width=.95\textwidth]{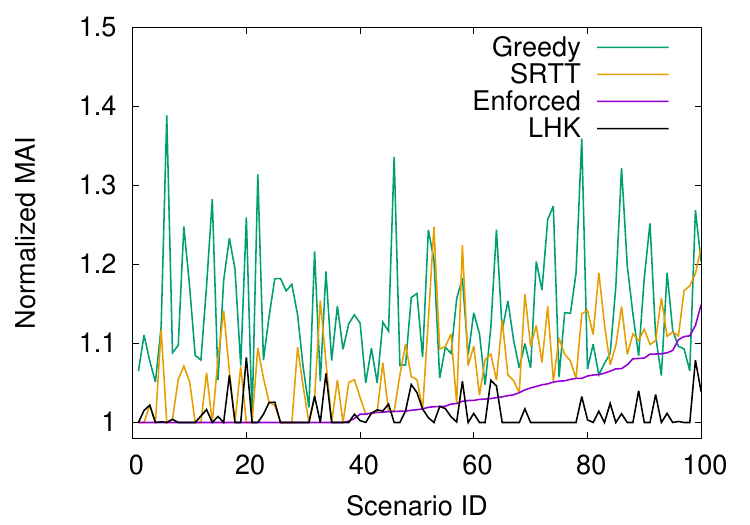}
		\caption{Grid}
	\end{subfigure}%
	\begin{subfigure}{0.32\textwidth}
		\includegraphics[width=.95\textwidth]{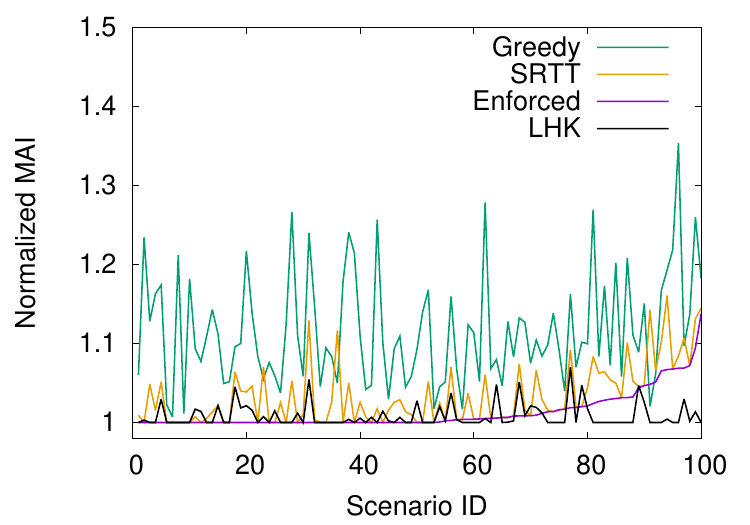}
		\caption{Cluster}
	\end{subfigure}%
	\begin{subfigure}{0.32\textwidth}
		\includegraphics[width=.95\textwidth]{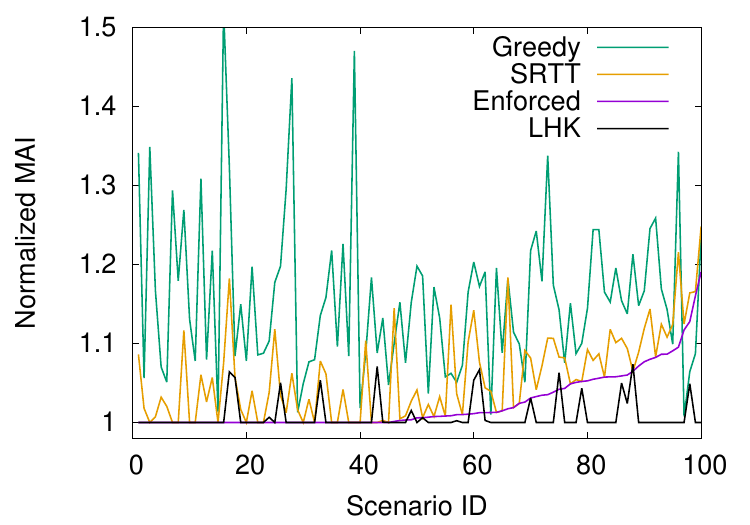}
		\caption{Outlier}
	\end{subfigure}
	\caption{MAIs of the routes generated by different approaches in the \emph{8-node} scenarios, normalized to the optimal solution}	
	\label{fig:MAIs_small}
\end{figure*}

\begin{figure*}[tb]
  \centering
	\begin{subfigure}{0.32\textwidth}
		\includegraphics[width=.95\textwidth]{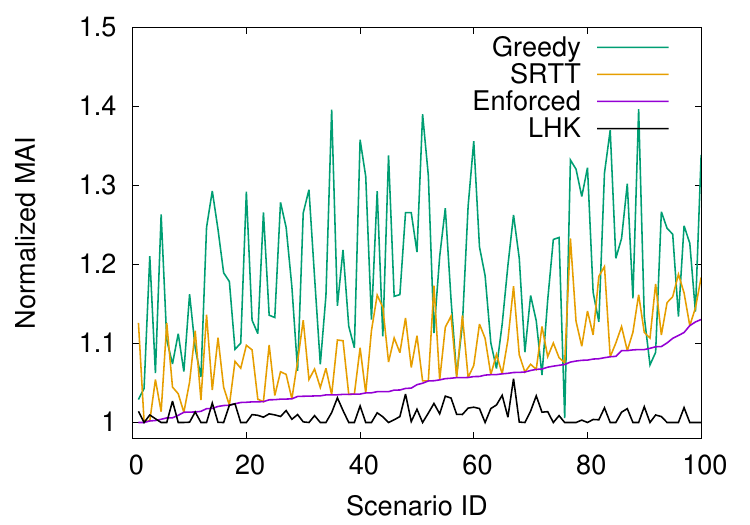}		
		\caption{Grid}
	\end{subfigure}%
	\begin{subfigure}{0.32\textwidth}
		\includegraphics[width=.95\textwidth]{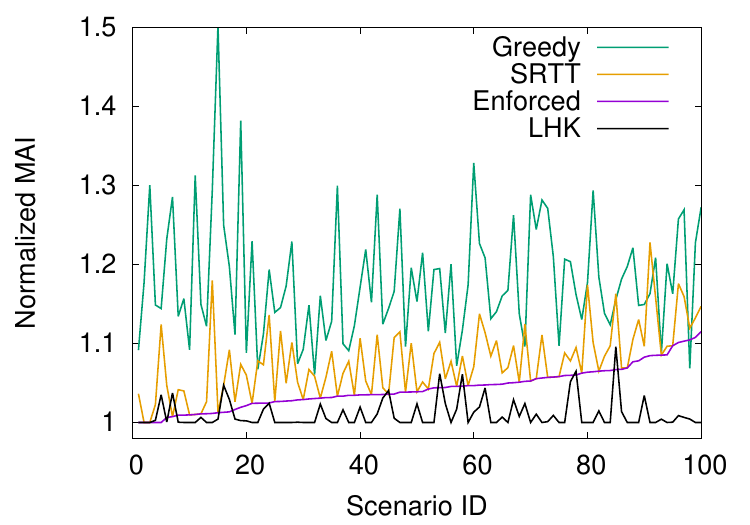}
		\caption{Cluster}
	\end{subfigure}%
	\begin{subfigure}{0.32\textwidth}
		\includegraphics[width=.95\textwidth]{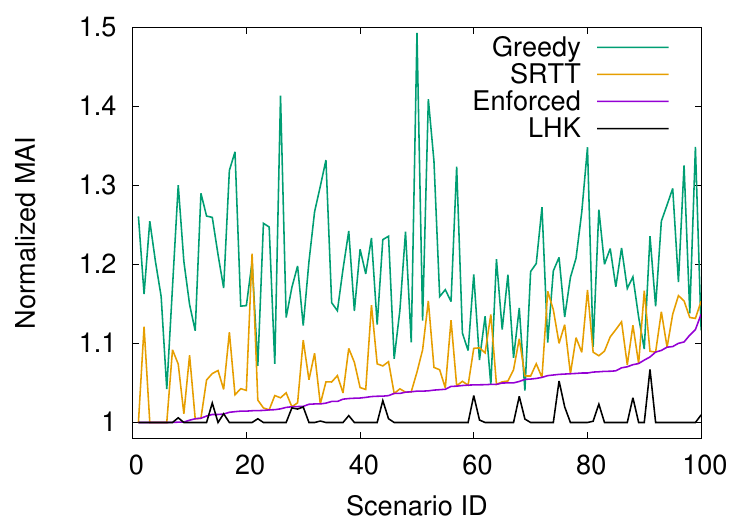}
		\caption{Outlier}
	\end{subfigure}
	\caption{MAIs of the routes generated by different approaches in the \emph{20-node} scenarios, normalized to the optimal solution.}
	\label{fig:MAIs_medium}
	\vspace{-1em}
\end{figure*}

\begin{table}[tb]
	\centering
	\caption{Average normalized MAIs of routes generated by different approaches across various scenario setups.}
	\label{tab:avgMAI}
	\begin{tabular}{c|c|cccc}
		\toprule
		& \textbf{Scenario} & \textbf{Greedy} & \textbf{SRTT} & \textbf{Enforced} & \textbf{LHK}\\
		\midrule
	 	\multirow{3}{*}{\emph{8-node}} &	 \emph{Grid} & 1.141		& 1.075		& 1.030		& 1.011 \\
		& \emph{Cluster} & 1.116	& 1.034	& 1.013	& 1.008 \\
		& \emph{Outlier} & 1.157	& 1.059	& 1.024	& 1.008 \\
		\midrule
		\multirow{3}{*}{\emph{20-node}} & \emph{Grid} & 1.194		& 1.093		& 1.052		& 1.009 \\
		& \emph{Cluster} & 1.181	& 1.076	& 1.043	& 1.010 \\
		& \emph{Outlier} & 1.198	& 1.076	& 1.042	& 1.004 \\
		\bottomrule
	\end{tabular}
	\vspace{-1em}	
\end{table}

We evaluate the MAI of routes generated by different approaches in synthetic scenarios.
To account for variations in MAI across different data node distributions, we normalize the MAI values relative to the optimal MAI, i.e., the MAI of the route generated by the \emph{Dynamic Programming} (DP) approach, for each scenario.
A normalized MAI of $1$ indicates that the route achieves the same MAI as the optimal route.
Note that \emph{DP} is not feasible for scenarios with a large number of nodes due to its high computational and memory requirements.

Figure~\ref{fig:MAIs_small} and Figure~\ref{fig:MAIs_medium} depict the normalized MAIs for their respective scenario setups, each containing $100$ scenarios.
The scenarios are sorted based on the normalized MAI of the routes generated by \emph{Enforced}.
Among the approaches, \emph{LHK} and \emph{Enforced} tend to generate routes with the lowest MAIs, with a maximum of $1.2$ and $1.15$ times the optimal MAI in the \emph{8-node} and \emph{20-node} scenarios, respectively.
In contrast, \emph{SRTT} and \emph{Greedy} generally produce routes with higher MAIs.
Table~\ref{tab:avgMAI} summarizes the average normalized MAIs of routes generated by different approaches across various scenario setups.

Although \emph{LHK} has the least average normalized MAI, there are scenarios where \emph{Enforced} generates routes with lower MAIs.
For example, in $47$ out of $100$ \emph{8-node} scenarios with a \emph{Grid} distribution, \emph{Enforced} generates routes with MAIs no larger than those of \emph{LHK}.
With a larger number of nodes in the \emph{20-node} scenarios, the likelihood of \emph{Enforced} outperforming \emph{LHK} decreases.
However, in $7$ out of $100$ scenarios with a \emph{Grid} distribution, \emph{Enforced} still performs at least as well as \emph{LKH}.
Nonetheless, \emph{Enforced} achieves the optimal MAI in at least $37\%$ of the \emph{8-node} scenarios, regardless of node distribution.
These empirical results support our claim that minimizing round trip time does not necessarily lead to a smaller MAI, highlighting \emph{Enforced} as a competitive approach for optimizing MAI in patrolling routes.

\begin{table}[tb]
	\centering
	\caption{Number of scenarios (out of $100$) where each approach produces the route with the optimal MAI.}
	\label{tab:LHKvsEnforced}
	\begin{tabular}{c|c|cc|c}
		\toprule
		& \textbf{Scenario} & \textbf{Enforced} & \textbf{LHK} & \textbf{Hybrid}\\
		\midrule
		\multirow{3}{*}{\emph{8-node}} &	 \emph{Grid} & 37 & 49 & 66 \\
		&	 \emph{Cluster} & 53 & 50 & 74 \\
		&	 \emph{Outlier} & 45 & 78 & 86 \\
		\midrule
		\multirow{3}{*}{\emph{20-node}} &	 \emph{Grid} & 2 & 28 & 30 \\
		&	 \emph{Cluster} & 5 & 48 & 50 \\
		&	 \emph{Outlier} & 6 & 78 & 78 \\
		\bottomrule
	\end{tabular}
	\vspace{-1em}
\end{table}

\noindent \textbf{Hybrid Approach}:
Based on the evaluation results, we propose a hybrid approach that leverages the strengths of both \emph{Enforced} and \emph{LHK}.
The hybrid approach generates routes by applying both \emph{Enforced} and \emph{LHK}, then selects the route with the smallest MAI among the two.
Table~\ref{tab:LHKvsEnforced} summarizes the number of scenarios in which each of the three approaches generates the route with an MAI equal to that of the optimal route.

These results show that the hybrid approach achieves the optimal MAI in at least $66\%$ and $30\%$ of the \emph{8-node} and \emph{20-node} scenarios, respectively, regardless of the node distribution.
Note that except in the \emph{20-node} scenarios with \emph{Outlier} distribution, the hybrid approach leads to more scenarios with the optimal MAI compared to \emph{LHK}, due to the fact that in some scenarios, \emph{Enforced} generates routes with lower MAIs than \emph{LHK}.

\subsection{Computation Time}\label{sec:eval_Exec}

\begin{figure}[tb]
	\centering
	\includegraphics[width=0.95\linewidth]{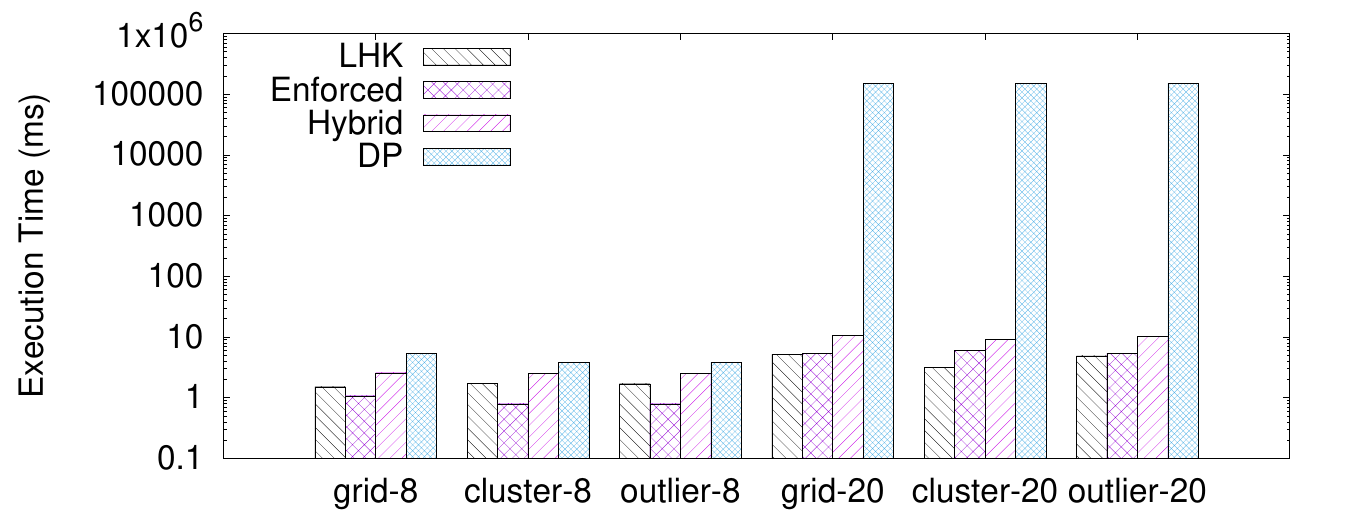}
	\caption{Average computation time for route generation across different approaches.}
	\label{fig:computation_time}
	\vspace{-1em}
\end{figure}

We compare the average computation times required to generate routes using different approaches across various scenario setups, focusing on \emph{Enforced}, \emph{LHK}, \emph{DP}, and \emph{Hybrid}.
The computation time for route generation is critical in practice, as data nodes may be redeployed, triggering updates to the route to ensure data freshness.
The analysis focuses solely on the time needed to generate the routes, excluding the time for calculating travel times between nodes and generating outputs.
Computation time is measured in milliseconds (ms).

Figure~\ref{fig:computation_time} illustrates the average computation time for route generation by each approach, with the y-axis on a logarithmic scale.
We observe that the computation time increases with the number of nodes in a scenario, while the distribution of data nodes has no significant effect.
The \emph{DP} approach exhibits the longest computation time, which grows exponentially with the number of nodes due to its $O(2^N \cdot N^3)$ time complexity, as discussed in Section~\ref{sec:edge_enforcement}, where $N$ is the number of nodes.
With a larger node counts, such as $24$ nodes in our implementation, \emph{DP} becomes impractical and fails to provide a solution within a reasonable time frame.
Specifically, a scenario with $23$ nodes takes approximately $30$ minutes to generate the optimal route, while a scenario with $24$ nodes is infeasible.
 
In contrast, the other three approaches generate solutions within a few milliseconds.
\emph{LHK} has a slightly longer computation time than \emph{Enforced} in the \emph{8-node} scenarios, but \emph{Enforced} takes longer as the number of nodes increases.
The \emph{Hybrid} method, as the combination of \emph{Enforced} and \emph{LHK}, has a computation time that is the sum of the two individual approaches.
Still, compared to \emph{DP}, the computation time of \emph{Hybrid} is significantly shorter, making it a practical approach for generating routes with minimized MAI.

In summary, while \emph{DP} yields an optimal route with the minimized MAI, it requires significantly more computation time and memory than other approaches.
On the other hand, \emph{Enforced} and \emph{LHK} generate routes with competitive MAIs in a fraction of the time required by \emph{DP}.
By combining the strengths of \emph{Enforced} and \emph{LHK}, the \emph{Hybrid} approach can achieve the optimal MAI in most scenarios in a reasonable computation time, making it a practical solution for minimizing the MAI of patrolling routes.

\section{Conclusion}\label{sec:conclusion}

In this work, we address the \emph{AoI-aware route planning problem}, which aims to minimize the Maximum Age-of-Information (MAI) in systems employing the \emph{Store-Carry-Forward} (SCF)  paradigm.
We prove that the decision version of the problem is an NP-Complete problem.
To tackle this challenge, we propose two approximation approaches, \emph{Shortest Round Trip Time} and \emph{Edge Enforcement}, along with a theoretical analysis of their approximation ratios.
The evaluation results confirm our claims regarding round trip time and MAI while also validating our theoretical analysis.
By combining the \emph{Enforced} approach with the state-of-the-art TSP solver, we develop a hybrid approach that effectively generates routes with near-optimal MAI, demonstrating its practical efficiency.

\bibliographystyle{IEEEtran}
\bibliography{ref}

\end{document}